\documentclass[11pt]{article}
\usepackage[margin=1in]{geometry}
\usepackage{amsfonts}
\usepackage{amsmath,amsthm,amssymb}
\usepackage{latexsym}
\usepackage{epic}
\usepackage{epsfig}
\usepackage{hyperref}
\usepackage{verbatim}
\usepackage[justification=centering]{caption}
\usepackage{enumitem}
\usepackage{color}
\allowdisplaybreaks[1]
\usepackage[numbers]{natbib}
\usepackage{algorithm}
\usepackage{algorithmic}
\usepackage{setspace}

\newtheorem{theorem}{Theorem}[section]

\newtheorem{lemma}[theorem]{Lemma}
\newtheorem{proposition}[theorem]{Proposition}

\newtheorem{definition}[theorem]{Definition}

\newtheorem{fact}[theorem]{Fact}
\newtheorem{example}[theorem]{Example}

\newtheorem{prop}[theorem]{Proposition}

\newtheorem*{conjecture*}{Conjecture}
\newtheoremstyle{nonindented}{1ex}{1ex}{}{}{\bfseries}{.}{.5em}{}
\newtheoremstyle{indented}{1ex}{1ex}{\itshape\addtolength{\leftskip}{0.6cm}\addtolength{\rightskip}{0.6cm}}{}{\bfseries}{.}{.5em}{}
\theoremstyle{nonindented}
\theoremstyle{indented}
\theoremstyle{plain}

\newcommand{\set}[1]{\left\{ #1 \right\}}
\newcommand{\union}{\cup}
\newcommand{\intersect}{\cap}

\newcommand{\sm}{\setminus}

\renewcommand{\bar}{\overline}

\def\Pr{\qopname\relax n{\mathbf{Pr}}}
\def\Ex{\qopname\relax n{\mathbf{E}}}

\newcommand{\RR}{\mathbb{R}}
\newcommand{\RRp}{\RR_+}

\newcommand{\NN}{\mathbb{N}}

\def\A{\mathcal{A}}

\def\D{\mathcal{D}}
\def\E{\mathcal{E}}
\def\F{\mathcal{F}}

\def\I{\mathcal{I}}

\def\M{\mathcal{M}}

\def\P{\mathcal{P}}

\def\W{\mathcal{W}}

\def\Y{\mathcal{Y}}

\def\eps{\epsilon}
\def\sse{\subseteq}

\newcommand{\one}{{\bf 1}}

\newcommand{\eat}[1]{}

\newcommand{\mini}[1]{\mbox{minimize} & {#1} &\\}
\newcommand{\maxi}[1]{\mbox{maximize} & {#1 } & \\}

\newcommand{\st}{\mbox{subject to} }
\newcommand{\con}[1]{&#1 & \\}
\newcommand{\qcon}[2]{&#1, & \mbox{for } #2.  \\}

\newenvironment{lp*}{\begin{equation*}  \begin{array}{lll}}{\end{array}\end{equation*}}

\title{The Outer Limits of Contention Resolution on Matroids and Connections to the Secretary Problem}%\footnote{A note to readers familiar with the original version of this manuscript: In addition to a number of improvements and clarifications, the most notable addition here is a response to the major critique of the original version. Specifically we show, and we believe quite convincingly, that our most technical result (Theorem~\ref{thm:improving_uncontentious}) does not follow from prior work, neither directly nor without significant new ideas.}}% \\ PRELIMINARY DRAFT \\ DO NOT DISTRIBUTE}

\author{
Shaddin Dughmi\thanks{This work was supported by NSF CAREER Award CCF-1350900.} \\
Department of Computer Science\\
University of Southern California\\
{\tt shaddin@usc.edu}
}

\begin{document}

\maketitle

\begin{abstract}
Contention resolution schemes have proven to be a useful and unifying abstraction for a variety of constrained optimization problems, in both offline and online arrival models. Much of prior work restricts attention to product distributions for the input set of elements, and studies contention resolution for increasingly general packing constraints, both offline and online. In this paper, we instead focus on generalizing the input distribution, restricting attention to matroid constraints in both the offline and online random arrival models. In particular, we study contention resolution when the input set is arbitrarily distributed, and may exhibit positive and/or negative correlations between elements. We characterize the distributions for which offline contention resolution is possible, and establish some of their basic closure properties. Our characterization can be interpreted as a distributional generalization of the matroid covering theorem.  For the online random arrival model, we show that contention resolution is intimately tied to the secretary problem via two results. First, we show that a competitive algorithm for the matroid secretary problem implies that online contention resolution is essentially as powerful as offline contention resolution for matroids, so long as the algorithm is given the input distribution. Second, we reduce the matroid secretary problem to the design of an online contention resolution scheme of a particular form.
%%% Local Variables:
%%% mode: latex
%%% TeX-master: "ucrs"
%%% End:

\end{abstract}

\newcommand{\ind}{\mathbf{Ind}}
\renewcommand{\binom}{\mathbf{Binom}}
\newcommand{\spn}{\mathbf{span}}
\newcommand{\rank}{\mathbf{rank}}
\newcommand{\OPT}{\mathbf{OPT}}
\newcommand{\convexhull}{\mathbf{convexhull}}

\section{Introduction}
%Opening paragraph. CRS are a unifying paradigm for a bunch of different applications, offline and online.
The notion of a contention resolution scheme (CRS) abstracts a familiar task in constrained optimization: converting  a (random) set-valued solution which is \emph{ex-ante} (i.e., on average) feasible for a packing problem  to one which is \emph{ex-post} (i.e., always) feasible. Unlike randomized rounding algorithms more broadly, which in general may be catered to both the constraint and objective function at hand, a contention resolution scheme is specific only to the constraints of the problem, and preserves solution quality in a manner which is largely agnostic to the objective function\footnote{In its most general form, a CRS approximately preserves all linear objective functions simultaneously, whereas a \emph{monotone} CRS approximately preserves all submodular objectives \cite{CRS}.} --- element by element.  Since they were formalized by \citet{CRS}, CRSs have been connected to a variety of online and offline computational tasks, including rounding the solutions of mathematical programs \cite{CRS}, online mechanism design and stochastic probing \cite{OCRS,ROCRS}, and prophet inequalities \cite{OCRS, OCRS_prophet}.

%Prior work has mostly considered product distributions which are ex-ante feasible, and generalized the set system. Here, motivated by the secretary problem, we focus on matroids but non-product distributions and not ex-ante feasible. We push the boundary on the distribution instead. Our goal is to understand the limits of contention resolution, particularly when elements may be correlated.

Starting with \cite{CRS}, prior work defines an (offline) contention resolution scheme for a set system $(\E, \I)$ --- where $\E$ is a \emph{ground set} of \emph{elements} and $\I \sse 2^\E$ is a downwards-closed family of \emph{feasible sets} --- as an algorithm which takes as input the marginal probabilities $x \in [0,1]^\E$ of a product distribution $\D$ supported on $2^\E$ as well as a random set $R \sim \D$ of \emph{active elements}, and must output a feasible subset $S$ of $R$. The contention resolution scheme is $\alpha$-competitive if $\Pr [ i \in S ] \geq \frac{1}{\alpha} \Pr [ i \in R]$ holds for all product distributions of interest --- typically those with marginals $x$ in the convex hull of indicator vectors of $\I$ (ex-ante feasibility). In \emph{online} contention resolution schemes, first explored by \citet{OCRS} and subsequently by \citet{ROCRS} and \citet{OCRS_prophet}, the active elements $R$ arrive sequentially and the decision to include an element in $S$ must irrevocably be made online.

The existing literature on (offline and online) contention resolution has mostly restricted attention to ex-ante-feasible and given product distributions, and varied the set system (e.g. matroids, knapsacks, and their intersections), all the while drawing connections to applications such as stochastic online problems, approximation algorithms, mechanism design, and prophet inequalities. In this paper, we restrict our attention to matroid constraints,\footnote{Though some of our results hold beyond matroids; we discuss those in the conclusion section.} and instead focus on generalizing the class of input distributions. Our main goal is to understand the power and limitations of contention resolution, offline and online, in the presence of correlations in the input distribution and without regard to ex-ante feasibility. A secondary goal is to understand how knowledge of the distribution influences contention resolution. In pursuit of both goals, we draw connections between contention resolution and the secretary problem on matroids, shedding light on challenges posed by the matroid secretary conjecture in the process.

%We seek characterizations of uncontentious distributions. We also try to understand whether online matters, focusing on random arrival. We show two connections, one in each direction.

%Mention how some of our results hold beyond matroids, and we discuss it in appendix

\subsection*{Results}
Our first set of results develops an understanding of offline contention resolution on matroids. We begin with a characterization of the class of \emph{$\alpha$-uncontentious distributions}: those distributions $\D \in \Delta(2^\E)$ permitting $\alpha$-competitive offline contention resolution for a given matroid. Most notably, we show that a distribution is $\alpha$-uncontentious if and only if it satisfies a family of $2^{|\E|}$ inequalities, one for each subset of the ground set. Moreover, we observe that our inequality characterization is the natural generalization of the  \emph{matroid base covering theorem} (see e.g. \cite{welsh}) from covering a set of elements to covering a distribution over sets of elements. In other words, we show that \emph{contention resolution is the natural distributional generalization of base covering.} Leveraging our characterization, we establish some basic closure properties of the class of uncontentious distributions, and present some examples of uncontentious distributions exhibiting negative and positive correlation between elements. Finally, we examine whether knowledge of the distribution $\D$ is essential to contention resolution, and exhibit an impossibility result: any contention resolution scheme which has nontrivial guarantees for all $\alpha$-uncontentious distributions cannot be \emph{prior-independent}, in that it cannot make do with a finite number of samples from the distribution, even for very simple matroids.

Our second set of results concerns online contention resolution on matroids in the random arrival model, and in particular its connection to the matroid secretary problem. First, we show that a competitive secretary algorithm for a matroid implies that online contention resolution is essentially as powerful as offline contention resolution for that same matroid: a $\gamma$-competitive secretary algorithm implies that any $\alpha$-uncontentious distribution permits $\gamma \alpha$-competitive online contention resolution. %One interpretation of this result is that such ``universal'' online contention resolution is a necessary step en-route to proving the matroid secretary conjecture.

Second, we provide evidence that contention resolution might hold the key to resolving the matroid secretary conjecture. As our most technically-involved result, we leverage our characterization of uncontentious distributions to show that the random set of \emph{improving elements} in a weighted matroid --- as originally defined by \citet{karger_matroidsampling} --- is $O(1)$-uncontentious. Since the improving elements can be recognized online, and moreover hold a constant fraction of the weighted rank of the matroid in expectation, our result can be loosely interpreted as a reduction from the matroid secretary problem to online contention resolution for a particular uncontentious distribution. There is one major caveat to this interpretation of our result, however: not only does the set of active (improving) elements arrive online, but so does the description of the uncontentious distribution from which that set is drawn. Though we present our proof of this result in an elementary form, the underlying arguments are reminiscent of --- and inspired by --- those often encountered in the analysis of martingales: we condition on carefully-chosen random variables, and employ a delicate charging argument between different probability events. 

Third, in response to feedback on the previous version of this manuscript, we show that our aforementioned result --- that improving elements are uncontentious --- cannot be derived as a consequence of prior work.

%In particular, we identify a (random) set of elements --- the set of \emph{improving elements}, as originally defined by \citet{karger_matroidsampling} --- which can be recognized online, and holds a constant fraction of the optimum value in expectation, and most importantly  and which if subjected to a competitive online contention resolution scheme would yield a competitive algorithm for the matroid secretary problem. This result can be loosely interpreted as a reduction from the matroid secretary problem to contention resolution, with one major caveat: not only does the set of active (improving) elements arrive online, but so does the description of the uncontentious distribution from which that set is drawn.

\subsection*{Additional Discussion of Related Work}

\subsubsection*{Contention Resolution Schemes}
%CRS, OCRS, ROCRS
Contention resolution schemes were introduced by \citet{CRS}, motivated by the problem of maximizing a submodular function subject to packing constraints. In particular, offline CRS were used to transform a randomized rounding algorithm which respects the packing constraints ex-ante to one which respects them ex-post, at the cost of the competitive ratio of the CRS. Their focus --- like that of all related work prior to ours --- was on product input distributions, in which case the optimal competitive ratio of an offline CRS was shown to equal the worst-case \emph{correlation gap} (first studied by \cite{correlation_gap_journal,CCPV11}) of the weighted rank function associated with the packing constraint. The characterization result of \cite{CRS} result forms the basis for ours.

Online contention resolution was first studied by \citet{OCRS}, and applied to a number of online selection problems. They show that simple packing constraints --- such as matroids, knapsacks, and matchings --- permit constant competitive online contention resolution schemes even when elements arrive in an unknown and adversarial order. Moreover, they show how to combine competitive online schemes for different constraints in order to yield competitive online schemes for their intersection. \citet{OCRS_prophet} obtain optimal online CRS in both the known adversarial-order model as well as the random-arrival model.  \citet{ROCRS} restrict attention to the random-arrival model, and obtain a particularly elegant algorithm and argument based on martingales, as well as improved competitive ratios for intersections of matroids and knapsacks. 

\subsubsection*{Prophet Inequalities}
Contention resolution is intimately tied to \emph{prophet inequality} problems, also known as \emph{Bayesian online selection} problems. In the traditional model for these problems, independent real-valued random variables with known distributions arrive online in a known but adversarial order, and the goal is to select a subset of the variables with maximum sum, subject to a packing constraint. An $\alpha$-competitive algorithm for a Bayesian online selection problem is also referred to as a prophet inequality with ratio $\alpha$, for historical reasons. Krengel, Sucheston, and Garling \cite{KS77, KS78} proved the first (classical) single-choice prophet inequality with ratio $1/2$ for selecting a single variable (i.e., a 1-uniform matroid packing constraint). Motivated by applications in algorithmic mechanism design, more recent work (e.g. \cite{HKS07,Alaei14,CHMS10,Yan11}) pursued prophet inequalities for more general packing constraints. Of particular note is the work of \citet{matroid_prophet}, who proved an optimal prophet inequality with ratio $1/2$ for matroids. Also notable is polylogarithmic prophet inequality for general packing constraints due to \citet{Rub16}.  The (easier) variant of Bayesian online selection problems in which the variables arrive in a uniformly random order has also received recent interest, resulting in improved prophet inequalities for various packing constraints \cite{prophet_secretary_a, prophet_secretary_b, prophet_secretary_c}.

It was shown by \citet{OCRS} that an online CRS yields a prophet inequality with the same competitive ratio, and in the same arrival model. A weak converse is also true, as shown by \citet{OCRS_prophet}: a stronger form of prophet inequality --- in particular one which competes against the \emph{ex-ante relaxation} of the Bayesian online selection problem --- yields an online CRS with the same competitive ratio and in the same arrival model.

\subsubsection*{Beyond Known Product Distributions}

%Discuss any and all work which tried to handle correlations. E.g. see references and discussion in prophet_correlated. Negative correlation always OK, positive correlation barely studied. Studied by Rinott_samuelcahn 87 and samuel cahn 91
%We note that none of the prior work has considered non-product distributions, for either CRS or prophet. \cite{OCRS_prophet} and \cite{ROCRS} use correlated negatively-correlated distributions as benchmarks, but ultimately their analysis is for product.

The vast majority of work on contention resolution or prophet inequalities, and all such work discussed thus far, restricts attention to known product distributions, and crucially exploits the product structure and knowledge of the distribution. We note the few exceptions next.

\citet{RSC87} and \citet{SC91} show that the single-choice prophet inequality, and some slight generalizations, continue to hold for negatively dependent random variables. It is known \cite{HK_survey} that there is no single-choice prophet inequality with ratio better than the number of variables  in the presence of arbitrary positive correlation.  Moreover, we are unaware of any nontrivial positive results for a class of distributions exhibiting positive correlation, in either prophet inequality or contention resolution models. We note that whereas \cite{OCRS_prophet} and \cite{ROCRS} use specially-crafted correlated distributions as benchmarks, their results and techniques do not appear to shed light on contention resolution or prophet inequalities in the presence of correlation more generally.

Some work has relaxed  the requirement that the distributions be known in prophet inequality problems. \citet{prophet_limited} study prophet inequality problems  when only a single sample is given from each distribution, and obtain constant competitive ratios for a variety of constraints. \citet{wang_singlesample} obtains an optimal algorithm for the single-choice prophet inequality, with ratio $1/2$, in the same single-sample model. \citet{prophet_iid_unknown} study the single-choice prophet inequality with i.i.d. variables drawn from an unknown distribution, and characterize the relationship between the competitive ratio and the number of samples available from the distribution.

\subsubsection*{Secretary Problems}
%Matsec: original, survey, svensson and lachish for state of the art.
In a \emph{generalized secretary problem}, a set of adversarially chosen variables arrive online in a random order, and the goal is to select a subset of the variables with maximum sum subject to a packing constraint. The (classical) single-choice \emph{secretary problem}, corresponding to a 1-uniform matroid constraint,  was introduced and solved by \citet{Dynkin}.  The \emph{matroid secretary problem} was introduced by \citet{matsec}, and has since spawned a long line of work. Constant-competitive algorithms have been discovered for most natural matroids and for some alternative models -- see \citet{dinitz_survey} for a semi-recent survey --- though the general conjecture remains open. The state of the art for the general matroid secretary problem is a $O(\log \log \rank)$-competitive algorithm due to \citet{lachish_loglog}, which was henceforth simplified by \citet{svensson_loglog}. 

%We note that the set of improving elements was defined by \citet{karger_matroidsampling}, in pursuit of faster algorithms for offline optimization problems on matroids.
%TODO: Mention secretary -> prophet secretary inequality (not ex ante)?
%TODO: Mention we are the first to connect secretary to CRS? Not clear we are or if its interesting to say so.

%%% Local Variables:
%%% mode: latex
%%% TeX-master: "ucrs"
%%% End:

\section{Preliminaries}

\subsection{Matroid Theory Basics}

We use standard definitions from matroid theory; for details see \cite{oxley,welsh}. A matroid $\M=(\E,\I)$ consists of a \emph{ground set} $\E$ of \emph{elements}, and a family $\I \sse 2^\E$ of \emph{independent sets}, satisfying the three \emph{matroid axioms}. A \emph{weighted matroid} $(\M,w)$ consists of a matroid $\M=(\E,\I)$ together with weights $w \in \RR^\E$ on the elements. We use the standard notions of a \emph{dependent set}, \emph{circuit}, \emph{flat}, and \emph{minor} in a matroid. We denote the \emph{rank} of a matroid $\M$ as $\rank(\M)$, and the rank of a set of elements $A$ in $\M$ as $\rank_\M(A)$, or $\rank(A)$ when $\M$ is clear from context. Overloading notation, we use $\rank^\M_w(A)$ to denote the \emph{weighted rank} of a set $A$ --- the maximum weight of an independent subset of $A$ --- in the weighted matroid $(\M,w)$, though we omit the superscript $\M$ when the matroid is clear from context. We note that both rank and weighted rank are submodular set functions on the ground set of the matroid. For $\M=(\E,\I)$ and $A \sse \E$, we denote the \emph{restriction} of $\M$ to $A$ as $\M | A$, \emph{deletion} of $A$  as $\M \sm A$, and \emph{contraction} by $A$ as $\M / A$.

When $\E$ is clear from context, and $S \sse \E$, we use $\one_S \in \set{0,1}^\E$ to denote the vector indicating membership in $S$. We often reference the \emph{matroid polytope} $\P(\M)$ of a matroid $\M=(\E,\I)$, defined as the convex hull of $\set{\one_S : S \in \I}$, or equivalently as the family of $x \in [0,1]^\E$ satisfying $\sum_{i \in S} x_i \leq \rank_\M(S)$ for all $S \sse \E$.

Throughout this paper we assume that any weighted matroid has distinct weights. This assumption is made merely to simplify some of our proofs, and --- using standard tie-breaking arguments --- can be shown to be without loss of generality in as much as our results are concerned. Under this assumption, we define $\OPT^\M_w(A)$ as the (unique) maximum-weight independent subset of $A$ of minimum cardinality (excluding zero-weight elements), and we omit the superscript when the matroid is clear from context.

% \begin{lemma}
%   Fix a matroid $\M=(\E,\I)$ with weights $w \succ 0$. For any pair of sets $A,B \sse \E$, we have $\OPT_w(A \union B) \sse \OPT_w(A) \union \OPT_w(B)$
% \end{lemma}
% \begin{proof}
%   Recall from the greedy algorithm that, for any $C \sse \E$, $\OPT_wC)$ includes precisely the elements of $C$ which are not spanned by higher weight elements of $C$. Take $i \in OPT(A \union B)$, and suppose $i \in A$. It follows that element $i$ is not spanned elements of $A \union B$ with weight higher than its own, and therefore $i$ is also not spanned by elements of $A$ with weight higher than its own, implying that $i \in \OPT_w(A)$. A similar argument applies to $i \in B$.
% \end{proof}

\subsection{The Matroid Secretary Problem}

%\begin{definition}[See \cite{matsec}]
  In the \emph{matroid secretary problem}, originally defined by \cite{matsec} there is matroid $\M=(\E,\I)$ with nonnegative weights $w: \E \to \RRp$ on the elements. Elements $\E$ arrive online in a uniformly random order $\Pi$, and an online algorithm must irrevocably accept or reject an element when it arrives, subject to accepting an independent set of $\M$. The algorithm is given $\M$ at the outset (as an independence oracle), but the weights $w$ are chosen adversarially before the order $\Pi$ is drawn and then are revealed online. The goal of the online algorithm is to maximize the expected weight of the accepted set of elements. 
We say that  an algorithm is \emph{$\alpha$-competitive} for a class of matroids if for every matroid $\M$ in that class and every adversarial choice of $w$, the expected weight of the accepted set (over the random choice of $\Pi$ and any internal randomness of the algorithm) is at least an $\alpha$ fraction of the maximum weight of an independent set of $(\M,w)$.
%\end{definition}

The \emph{matroid secretary conjecture}, posed by \cite{matsec}, postulates that the matroid secretary problem admits an (online) algorithm which is constant-competitive for all matroids.

\subsection{Miscellaneous Notation and Terminology}

We denote the natural numbers by $\NN$, and the nonnegative real numbers by $\RRp$. Given a set $\A$ with weights $w \in \RR^\A$, and a subset $B \sse \A$, we use the shorthand $w(B) = \sum_{i \in B} w_i$. We use $[n]$ as shorthand for the set ${1,\ldots,n}$. For a set $A$, we use $\Delta(A)$ to denote the family of distributions over $A$, and $2^A$ to denote the family of subsets of $A$.

Let $\A$ be a finite ground set. For a distribution $\D$ supported on $2^\A$, we define the vector $x \in [0,1]^\A$ of \emph{marginals} of $\D$ by $x_i = \Pr_{B \sim \D}[ i \in B]$, and refer to $x_i$ as the \emph{marginal probability} of $i$ in $\D$. When marginals $x \in [0,1]^\A$ are given, we use $\ind(x)$ to denote the distribution of the random set $B \sse \A$ which includes each element $i \in \A$ independently with probability $x_i$. We also use $\ind_p(\A)$ as shorthand for $\ind(x)$ when $x_i =p$ for all $i \in \A$.

% \begin{definition}
%   Given a finite set $\A$ and a distribution $\D$ supported on $2^\A$, we define the vector $x \in [0,1]^\A$ of \emph{marginals} of $\D$ by $x_i = \Pr_{B \sim \D}[ i \in B]$. We refer to $x_i$ as the \emph{marginal probability} of $i$ in $\D$.
% \end{definition}

% \begin{definition}
%   Given a finite set $\A$ and $p \in [0,1]$, let $\ind_p(\A)$ be the distribution of the random set $B \sse \A$ which includes each element of $\A$ independently with probability $p$. Equivalently, $\ind_p(\A)$ is the product distribution over subsets of $\A$ with all marginal probabilities equal to $p$.
% \end{definition}

% \begin{definition}
%   Given a finite set $\A$ and a vector $x \in [0,1]^\A$, let $\ind(x)$ be the distribution of the random set $B \sse \A$ which includes each element $i \in \A$ independently with probability $x_i$. Equivalently, $\ind(x)$ is the product distribution over subsets of $\A$ with marginals $x$.
% \end{definition}

% Note that when $B \sim \ind_p(A)$, we have $|B| \sim \binom(|A|,p)$. The following fact provides an alternate way to sample from $\ind_p(A)$
% \begin{proposition}\label{prop:sample_online}
%   Let $A$ be a finite set. The following process samples a set $B \sim \ind_p(A)$
%   \begin{itemize}
%   \item Sample $k \sim \binom(|A|, p)$
%   \item Let $B$ be a uniformly random element of $A \choose k$.
%   \end{itemize}
% \end{proposition}
% \begin{proof}
%   \todo{Well known and easy. Include for completeness.}
% \end{proof}

%%% Local Variables:
%%% mode: latex
%%% TeX-master: "ucrs"
%%% End:

\section{Understanding Contention Resolution}
%Define Offline, Online, Universal, Oblivious schemes

\subsection{The Basics of Contention Resolution}

The definitions below are parametrized by a given matroid $\M=(\E, \I)$.

\begin{definition}
A \emph{contention resolution map (CRM)} $\phi$ is a randomized function from $2^\E$ to $\I$ with the property that $\phi(R) \sse R$ for all $R \sse \E$. Such a map is \emph{$\alpha$-competitive} for a distribution $\D \in \Delta(2^\E)$  if, for $R \sim \D$, we have $\Pr[ i \in R] \leq \alpha \Pr [ i \in \phi(R)]$ for all $i \in \E$. 
\end{definition}

% \begin{definition}
%  A \emph{contention resolution scheme (CRS)} $\Phi$ for a class of distributions in $\Delta(2^\E)$ defines, for each $\D$ in that class, a contention resolution map $\phi_\D$. We say $\Phi$ is $\alpha$-competitive if each of its maps $\phi_\D$ is $\alpha$-competitive for its respective distribution $\D$.  

%   Let $\DD \sse \Delta(2^\E)$ be a family of distributions over subsets of the ground set $\E$. A \emph{contention resolution scheme (CRS)} $\Phi$ for $\DD$ is an algorithm that takes as input a description of a distribution $\D \in \DD$, and then computes a class of distributions in $\Delta(2^\E)$ defines, for each $\D$ in that class, a contention resolution map $\phi_\D$. We say $\Phi$ is an  \emph{$\alpha$-competitive CRS} (or an \emph{$\alpha$-CRS} for short) if each of its maps $\phi_\D$ is $\alpha$-competitive for its respective distribution $\D$. 
% \end{definition}

The following is known from \citet{CRS}. 
\begin{theorem}[\cite{CRS}]
  Every product distribution with marginals in $\P(\M)$ admits an $\frac{e}{e-1}$-competitive CRM.
\end{theorem}

\begin{definition}
  An \emph{online random-order contention resolution map} (henceforth \emph{online CRM} for short) is a contention resolution map $\phi$ which can be implemented as an algorithm in the online random-arrival model. In the online random-arrival model,  $\E$ is presented to the algorithm in a uniformly random order $(e_1,\ldots,e_n)$, and at the $i$th step the algorithm learns whether $e_i$ is \emph{active} --- i.e., $e_i \in R \sim \D$ --- and if so must make an irrevocable decision on whether to include $e_i$ in $\phi(R)$. 
%  An \emph{online random-order contention resolution scheme} (or \emph{online CRS} for short) $\Phi$ is a contention resolution scheme where each $\phi_\D$ is an online random-order CRM.
\end{definition}

The following is known from \citet{OCRS_prophet}. 

\begin{theorem}[\cite{OCRS_prophet}]
  Every product distribution with marginals in $\P(\M)$ admits a $\frac{e}{e-1}$-competitive online CRM.
\end{theorem}

\subsection{Uncontentious Distributions and their Characterization}

As shorthand, we refer to distributions which permit competitive (offline) CRMs as \emph{uncontentious}.
\begin{definition}
  Fix a matroid $\M=(\E,\I)$. For $\alpha \geq 1$, we say that a distribution $\D \in \Delta(2^\E)$ is \emph{$\alpha$-uncontentious} if it admits an $\alpha$-competitive contention resolution map. 
\end{definition}
For convenience, we also refer to a random set $R \sim \D$ as $\alpha$-uncontentious if its distribution $\D$ is $\alpha$-uncontentious. 
We prove the following characterization of uncontentious distributions.

\begin{theorem}\label{thm:characterize_uncontentious}
  Fix a matroid $\M=(\E,\I)$, and let $\D$ be a distribution supported on $2^\E$. The following are equivalent for every $\alpha \geq 1$.
  \begin{enumerate}[label=(\alph*)]
  \item $\D$ is $\alpha$-uncontentious (i.e., admits an $\alpha$-competitive contention resolution map).
  \item For every weight vector $w \in \RRp^\E$, the following holds for $R\sim \D$:
\[ \Ex [\rank_w(R)] \geq  \frac{1}{\alpha} \Ex[w(R)]\]
  \item For every $\F \sse \E$, the following holds for $R \sim \D$:
    \[ \Ex[ |R \intersect \F|] \leq \alpha \Ex [\rank(R \intersect \F)] \]
 % \item There exists a measure $\mu$ over $\I$, with total magnitude at most $\alpha$, such that $\sum_{S \in \I} \mu(S) \one_S \succeq x$.
  \end{enumerate}
\end{theorem}
\begin{proof}

Property (a) implies property (c) by applying an $\alpha$-CRM $\phi$ to $R$, noting that $\phi(R) \intersect \F$ is necessarily  an independent subset of $R \intersect \F$.
\begin{align*}
\Ex[\rank(R \intersect \F)] &\geq \Ex[ |\phi(R) \intersect \F |] \\
&= \sum_{i \in \F} \Pr[i \in \phi(R)] \\
&\geq \frac{1}{\alpha} \sum_{i \in \F} \Pr[i \in R] \\
&= \frac{1}{\alpha} \Ex[ |R \intersect \F |].
\end{align*}

%   (a) implies (b) by applying an $\alpha$-CRM $\phi$ to $R$: \[\Ex[ \rank_w(R)] \geq \Ex [ w(\phi(R))] = \sum_{i \in \E} w_i \pr[ i \in \phi(R)] \geq  \sum_{i \in \E} w_i \frac{\pr[ i \in R]}{\alpha} = \frac{1}{\alpha} \sum_{i \in \E} w_i x_i\]

% (b) implies (c) by using the indicator vector $w=\one_\F$: \[\Ex [ | R \intersect \F|] = \Ex [ w(R)] = \sum_{i \in \E} w_i x_i \leq \alpha \Ex [\rank_w(R)] = \alpha \Ex[\rank(R \intersect \F)]\]

Property (c) implies property (b) by a summation argument. Sort and number the elements $\E = (e_1,\ldots, e_n)$ in decreasing order of weights $w_1 \geq w_2 \geq \ldots \geq w_n \geq 0$, where $w_i$ denotes the weight of $e_i$. Denote $\E_i = \set{e_1,\ldots,e_i}$, and let $\E_0=\emptyset$, and $w_{n+1}=0$. Recalling that the greedy algorithm computes the maximum weight independent subset of a matroid, we get:
\begin{align*}
  \Ex [\rank_w(R)] &= \Ex\left[ \sum_{i=1}^n w_i \left( \rank(R \intersect \E_i) - \rank(R \intersect \E_{i-1}) \right) \right] &\mbox{Invoking the greedy algorithm on $\M | R$} \\
&= \Ex\left[ \sum_{i=1}^n (w_i - w_{i+1}) \rank(R \intersect \E_i)\right] &\mbox{Reversing order of summation}  \\
                   &\geq \frac{1}{\alpha} \Ex\left[ \sum_{i=1}^n (w_i - w_{i+1}) |R \intersect \E_i|\right] &\mbox{Invoking (c) and linearity of expectations} \\
&= \frac{1}{\alpha} \Ex\left[ \sum_{i=1}^n w_i \left( |R \intersect \E_i| - |R \intersect \E_{i-1}| \right) \right]  &\mbox{Reversing order of summation}. \\
%&= \frac{1}{\alpha}  \sum_{i=1}^n w_i \Pr[ i \in R]  &\mbox{By linearity of expectations}. \\
&= \frac{1}{\alpha} \Ex [ w(R)]. %&\mbox{By linearity of expectations}. 
%&= \frac{1}{\alpha} \sum_{i \in \E} w_i x_i
\end{align*}

Property (b) implies property (a) by a duality argument identical to that presented in \cite{CRS}. We present a self-contained proof here. Let $x=x(\D) \in [0,1]^\E$ denote the marginals of $\D$.  The distribution $\D$ is $\alpha$-uncontentious if the optimal value of the following LP, with variables $\beta$ and  $\lambda_\phi$ for each deterministic CRM $\phi$, is at least $\frac{1}{\alpha}$.
\begin{lp*}\label{lp:crs}
\maxi{\beta}
\st
\qcon{\sum_{\phi} \lambda_\phi \Pr_{R \sim \D}[i \in \phi(R)] \geq \beta  x_i}{i \in \E}
\con{\sum_{\phi} \lambda_\phi = 1}
\con{ \lambda \succeq 0}
\end{lp*}
The dual of the preceding LP is the following
\begin{lp*}
  \mini{\mu}
  \st
\qcon{\sum_{i \in \E} w_i \Pr_{R \sim D} [i \in \phi(R)] \leq \mu}{\mbox{all CRM $\phi$}}
\con{\sum_{i \in \E} w_i x_i = 1}
\con{w \succeq 0}
\end{lp*}
It is not hard to see that, at optimality, the binding constraint on $\mu$ corresponds to the CRM $\phi$ which maps each set $R$ to its maximum weight independent subset according to weights $w$. It follows that the optimal value of the dual, and hence the primal, equals the minimum over all weight vectors $w \succeq 0$ of the ratio $\frac{\Ex[\rank_w(R)]}{\sum_i w_i x_i}$.  (b) implies that this quantity is at least $\frac{1}{\alpha}$, as needed.

% Finally, we show that (d) and (a) are equivalent. Recall that $\D$ is $\alpha$-uncontentious if and only LP \eqref{lp:crs} has value at least $\frac{1}{\alpha}$. It is easy to see that the value of LP \eqref{lp:crs} is the inverse of the value of the following LP which optimizes over measures.

% \begin{lp*}
% \mini{\gamma}
% \st
% \qcon{\sum_{\phi} \lambda_\phi \Pr_{R \sim \D}[i \in \phi(R)] \geq  x_i}{i \in \E}
% \con{\sum_{\phi} \lambda_\phi = 1}
% \con{ \lambda \succeq 0}
% \end{lp*}

\end{proof}

We note that the equivalence between (a) and (b) is essentially implicit in the arguments of \cite{CRS}. Condition (c) is the most notable part of Theorem~\ref{thm:characterize_uncontentious}, in no small part because it is  reminiscent of the \emph{matroid base covering theorem} (see e.g., \cite{welsh}). This theorem can equivalently be stated as follows: a (deterministic) set $T \sse \E$ in a matroid $\M=(\E,\I)$ can be \emph{covered by} (i.e., expressed as a union of) $\alpha \in \NN$ independent sets if and only if $|S| \leq \alpha \ \rank_\M(S)$ for all $S \sse T$. In light of part (c) of Theorem~\ref{thm:characterize_uncontentious}, a set $T$ of elements can be covered by $\alpha$ independent sets if and only if the point distribution on $T$ is $\alpha$-uncontentious. Therefore, \emph{we can interpret contention resolution as a distributional  generalization of base covering.}

\subsection{Elementary Properties of Uncontentious Distributions}
We state some elementary, yet quite useful, properties of uncontentious distributions.

\begin{prop}
  Fix a matroid $\M$. Every  $\alpha$-uncontentious distribution $\D$ for $\alpha \geq 1$ has marginals $x(\D) \in \alpha \P(\M)$.
\end{prop}
\begin{proof}
  Let $x=x(\D)$ and $R \sim \D$. From Theorem~\ref{thm:characterize_uncontentious} (c), for every set of ground set elements $\F$ we have \[\sum_{i \in \F} x_i = \Ex[|R \intersect \F|] \leq \alpha \Ex[\rank_\M(R \intersect \F)] \leq \alpha\ \rank_\M(\F).\]
These are the inequalities describing $\alpha \P(\M)$.
\end{proof}

\begin{prop}\label{prop:uncontentious_mixing}
  Fix a matroid. A mixture of $\alpha$-uncontentious distributions is $\alpha$-uncontentious.
\end{prop}
\begin{proof}
Follows directly from Theorem~\ref{thm:characterize_uncontentious} (b) and linearity of expectations.  
\end{proof}

\begin{prop}\label{prop:uncontentious_restriction}
  Fix a matroid $\M=(\E,\I)$, and let $\M=(\E',\I')$ be a minor of $\M$, with $\E' \sse \E$. If a random set $R \sse \E'$ is $\alpha$-uncontentious in $\M'$, then $R$ is also $\alpha$-uncontentious in $\M$.
\end{prop}
\begin{proof}
  An independent set of $\M'$ is also independent in $\M$. Therefore, the proposition follows by simply applying the same CRM in the context of the larger matroid $\M$. 
%  Follows directly from Theorem \ref{thm:characterize_uncontentious} (b) and the fact that the weighted rank function of $\M | \E'$ agrees with the weighted rank function of $\M$ on subsets of $\E'$.
\end{proof}

\begin{prop}\label{prop:uncontentious_subsampling}
  Fix a matroid. Let $R$ be an $\alpha$-uncontentious random set, and let $R' \sim \ind_p(R)$ for some $p \in [0,1]$. The random set $R'$ is $\alpha$-uncontentious as well.
\end{prop}
\begin{proof}
  We use Theorem~\ref{thm:characterize_uncontentious} (b). For any weight vector $w$, submodularity of the weighted rank function implies that $\Ex[ \rank_w(R') ] \geq p \Ex [ \rank_w(R)]$. It follows that $\Ex [w(R') ] = p \Ex [ w(R)] \leq p \alpha \Ex[\rank_w(R)] \leq \alpha \Ex[ \rank_w(R')]$.
\end{proof}

We note that Proposition ~\ref{prop:uncontentious_subsampling} is tight when both $p$ and $\alpha$ are absolute constants. In particular, the random set $R'$ cannot be guaranteed to be $\alpha'$-uncontentious for a constant $\alpha' < \alpha$, even if $p$ is a very small constant. To see this, consider the a $1$-uniform matroid with elements $[n]$, and the following $2$-uncontentious random set $R$: For every singleton $i \in [n]$ we have $\Pr[R=\set{i}] = \frac{1}{n+1}$, and $\Pr[R=[n]]= \frac{1}{n+1}$.

\subsection{Examples of Uncontentious Distributions}
We now present some examples of uncontentious distributions in order to develop a feel for them. As mentioned previously, and shown in \cite{CRS}, every product distribution with marginals in the matroid polytope is $\frac{e}{e-1}$-uncontentious.  Combined with Proposition \ref{prop:uncontentious_mixing}, this extends to mixtures of product distributions.
\begin{proposition}
  Fix a matroid $\M=(\E,\I)$, and let $\D \in \Delta(2^\E)$ be a mixture of product distributions, each with marginals in $\P(\M)$. It follows that $\D$ is $\frac{e}{e-1}$-uncontentious.
\end{proposition}

Going beyond product distributions and their mixtures, if a distribution satisfies a certain strong notion of negative correlation, defined in \cite{CVZ10}, then it also is $\frac{e}{e-1}$-uncontentious.
\begin{proposition}
 Fix a matroid $\M=(\E,\I)$, and let $\D \in \Delta(2^\E)$ be a distribution with marginals $x=x(D) \in \P(\M)$. Assume that $\D$  satisfies the \emph{property of increasing submodular expectations}: for every submodular function $f$ we have $\Ex_{R \sim \D} [f(R)] \geq \Ex_{S \sim \ind(x)} [f(S)]$.\footnote{In fact, it suffices for $\D$ to satisfy the (weaker) property of increasing expectations for matroid rank functions (or, equivalently, their weighted sums).} It follows that $\D$ is $\frac{e}{e-1}$-uncontentious.
\end{proposition}
\begin{proof}
  This is immediate by combining Theorem~\ref{thm:characterize_uncontentious} (b) with the property of increasing submodular expectations and the fact that $\ind(x)$ is $\frac{e}{e-1}$-uncontentious.
\end{proof}

As shown in \cite{CVZ10}, the property of increasing submodular expectations is stronger than the following standard notion of negative correlation for $R \sim \D$: For all sets $T$, $\Pr[T \sse R] \leq \prod_{i \in T} \Pr[ i \in R]$ and $\Pr[T \sse \bar{R}] \leq \prod_{i \in T} (1- \Pr[ i \in R])$.\footnote{A natural question is whether negative correlation suffices for the distribution to be $\frac{e}{e-1}$-uncontentious. This is open as far as we know.}  However, we can show that there are distributions exhibiting positive correlation which are also uncontentious for specific matroids. We now list some examples of uncontentious distributions exhibiting positive correlation.

\begin{example}
  Let $\M$ be a $k$-uniform matroid with $n$ elements where $2 \leq k \leq n$. Let the random set $R$ be empty with probability $1/2$, and a uniformly random base of $\M$ otherwise. 

It is clear that $R$ is $1$-uncontentious, since it is supported on the family of independent sets. However, for each distinct pair of elements $i$ and $j$, we have $\Pr [i \in R] = \Pr [j \in R] = \frac{k}{2n}$, yet $\Pr [ i \in R | j \in R] = \Pr [ j \in R | i \in R] = \frac{k-1}{n-1} > \frac{k}{2n}$.
\end{example}

The next example will feature repeatedly in this paper, since it is the random set of improving elements for the rank $1$ matroid.

\begin{example}\label{ex:1}
  Consider the $1$-uniform matroid with elements $[n]=\set{1,\ldots,n}$. For $k=0,1,\ldots,n-1$, let $R=\set{1,\ldots,k}$ with probability $2^{-(k+1)}$, and let $R=[n]$ with remaining probability $2^{-n}$. The random set $R$ is $2$-uncontentious, as evidenced by the CRM $\phi$ with $\phi(\set{1,\ldots,k})=\set{k}$ and $\phi(\emptyset)=\emptyset$, and a simple calculation. Note the positive correlation between elements $i < j$: $\Pr[j \in R] = 2^{-j}$, and $\Pr[ j \in R | i \in R] = 2^{i-j} > \Pr[ i \in R]$.
\end{example}
As a generalization of the previous example, we get the following.
\begin{example}
  Let $\M$ be a matroid with $m$ pairwise-disjoint bases $B_1,\ldots,B_m$. For each $k=1,\ldots,m-1$, let $R= \union_{i=1}^k B_i$ with probability $2^{-k}$, and let $R=\union_{i=1}^m B_m$ with the remaining probability $2^{1-m}$. The set $R$ is $2$-uncontentious, as evidenced by the CRM $\phi(\union_{i=1}^k B_i) = B_k$. However, for $e_i \in B_i$ and $e_j \in B_j$ with $i<j$, we have $\Pr[e_j \in R] = 2^{1-j}$ and  $\Pr [ e_j \in R | e_i \in R] = 2^{i-j} > \Pr[ e_j \in R]$.
\end{example}

\subsection{Contention Resolution Schemes, Universality, and Prior Dependence}

 \def\DD{\mathbb{D}}

A \emph{contention resolution scheme (CRS)} $\Phi$ for a matroid $\M=(\E,\I)$ and class of distributions $\DD \sse \Delta(2^\E)$ is an algorithm which takes as input a (possibly partial) description of a distribution $\D \in \DD$ and a sample $R \sim \D$, and outputs $S \in \I$ satisfying $S \sse R$. In effect, $\Phi$ is a collection of contention resolution maps $\phi_\D$, one for each $\D \in \DD$. In much of the prior work on contention resolution, $\DD$ was taken to be the class of product distributions with marginals in $\P(\M)$, and each $\D \in \DD$ is described completely via its marginals $x \in \P(\M)$. In such a setting, the notion of a CRS offers little beyond the notion of a CRM, as each distribution gets its own dedicated CRM. More generally, however, we allow $\DD$ to be an arbitrary class of distributions, and we allow the description to be partial and/or random; for example, $\D$ may be described by $m$ independent samples from $\D$.

Next, we set the stage by  defining some desirable contention resolution schemes, and establish some limitations on their existence. 

\begin{definition}
  Fix a matroid. For $\beta \geq \alpha > 1$, an \emph{$\alpha$-universal $\beta$-competitive CRS} is a CRS which is $\beta$-competitive for the class of $\alpha$-uncontentious distributions. %A \emph{constant-competitive universal CRS} is one which is $\alpha$-universal and $\beta$-competitive for some $\alpha \leq \beta = O(1)$.  %If the CRS in question can be implemented in the online random-arrival model, we say it is an \emph{online $\alpha$-Universal $\beta$-CRS}.
\end{definition}

By definition, there exists an (offline) $\alpha$-universal $\alpha$-competitive CRS for every $\alpha$ and every matroid. The notion of a universal scheme becomes more interesting when we restrict dependence on the prior, as per the following definitions. %We are unaware of any prior work on contention resolution beyond product distributions.

\begin{definition}
  Fix a matroid. A contention resolution scheme $\Phi$ is said to be \emph{prior-independent} if it is not given a complete description of $\D$ as input, but rather is given a set of independent samples from $\D$. When the number of samples is $m$, we say $\Phi$ is a \emph{prior-independent $m$-sample scheme}. The number of samples $m(\cdot)$ may be function of the size of the matroid. If $m=0$, we say the scheme is \emph{oblivious}: the scheme consists of a single contention resolution map.
\end{definition}

We now show that, if a scheme is universal, it cannot be prior-independent with any finite number of samples, even for very simple matroids. 
\begin{theorem}\label{thm:prior_independent_impossible}
  Let $\M$ be the 1-uniform matroid on $n$ elements. For every finite $m$, and every $1 < \alpha \leq \beta < n$,  there does not exist a $\beta$-competitive $\alpha$-universal  CRS for $\M$ which is prior independent with $m$ samples.
\end{theorem}

To prove Theorem~\ref{thm:prior_independent_impossible}, we first show that a prior-independent universal scheme implies the existence of an oblivious universal scheme; then we show that an oblivious universal scheme does not exist for the uniform matroid. This is captured in the two following lemmas.

\begin{lemma}
  Fix a matroid $\M$. If there exists a $\beta$-competitive $\alpha$-universal CRS $\Phi$ which is prior-independent with $m$ samples, then there exists an oblivious $\beta$-competitive $\alpha$-universal scheme $\Phi'$. %Moreover, if $\Phi$ is online then so is $\Phi'$. 
\end{lemma}
\begin{proof}
  Let $\D$ be any $\alpha$-uncontentious distribution. Let $\eps \in (0,1)$, and  let $\D'= \D'(\epsilon)$ be the mixture of $\D$ with the point distribution on the empty set with proportions $\epsilon$ and $1-\epsilon$ respectively. By Proposition \ref{prop:uncontentious_mixing} and the fact that the point distribution on the empty set is $1$-uncontentious, it follows that $\D'$ is $\alpha$-uncontentious. 

The CRS $\Phi$ induces a CRM $\phi_{\D'}$ on the distribution $\D'$, and by assumption $\phi_{\D'}$ is  $\beta$-competitive for $\D'$.  Since $\Phi$ is prior-independent with $m$ samples, its induced CRM  $\phi_{\D'}$ is a mixture over CRMs $\phi_S$, where $S=(S_1,\ldots,S_m)$ is a random vector of $m$ samples from $\D'$. With probability at least $(1-\epsilon)^m$, we have $S=\emptyset^m:=(\emptyset, \ldots, \emptyset)$. For $\phi_{\D'}$ to be $\beta$-competitive, in particular when with probability $\epsilon$ it is queried with a draw $R \sim D$,  a simple calculation shows that $\phi_{\emptyset^m}$ must be $\beta'$-competitive for $\D$ for $\beta'=\frac{(1-\epsilon)^m}{1/\beta + (1-\epsilon)^m - 1}$. As $\epsilon$ tends to $0$, $\beta'$ tends to $\beta$, and a basic analytic argument implies that $\phi_{\emptyset^m}$ is $\beta$-competitive for $\D$. Since $\D$ was chosen arbitrarily among $\alpha$-uncontentious distributions, and $\phi_{\emptyset^m}$ does not depend on $\D$, it follows that the oblivious scheme $\Phi'$ with $\phi'_\D = \phi_{\emptyset^m}$ for every $\D$ is $\beta$-competitive and $\alpha$-universal. \end{proof}

\begin{lemma}
The 1-uniform matroid with $n$ elements does not admit an oblivious $\beta$-competitive $\alpha$-universal CRS for any $1 < \alpha \leq \beta < n$.
\end{lemma}
\begin{proof}
  Let $[n]$ be the ground set of the matroid, and fix $\alpha$ such that $1 < \alpha < n$. An oblivious CRS consists of a single CRM $\phi$. There exists at least one element $i \in [n]$ such that $\Pr[ i \in \phi([n])] \leq 1/n$. Let $\eps= \alpha - 1 >0$, and consider the following random set $R$: For each $j \in [n] \sm i$ we have $R = \set{j}$ with probability $\frac{1}{n-1 + \epsilon}$, and $R= [n]$ with the remaining probability $\frac{\eps}{n-1 +\eps}$. The random set $R$ is $\alpha$-uncontentious: consider the CRM $\phi'$ with $\phi'(\set{j}) = j$ for $j \neq i$, and $\phi'([n]) = i$. However, our original CRM $\phi$ is no better than $n$-competitive for $R$, since its probability of selecting $i$ is no more than $\frac{1}{n} \Pr[R = [n]] = \frac{1}{n} \Pr[ i \in R]$.
\end{proof}

%%% Local Variables:
%%% mode: latex
%%% TeX-master: "ucrs"
%%% End:

\section{An Online Universal CRS from a Secretary Algorithm}

We show that competitive matroid secretary algorithms imply that every contention resolution scheme can be made online, in the random arrival model, without much loss.
\begin{theorem}\label{thm:sec_to_CRS}
 Suppose that there is a $\gamma$-competitive online algorithm for the secretary problem on matroid $\M$. It follows that every $\alpha$-uncontentious distribution admits an online $\gamma \alpha$-competitive contention resolution map. In other words, for every $\alpha$ there exists an online $\gamma \alpha$-competitive $\alpha$-universal contention resolution scheme for $\M$.
\end{theorem}
We interpret the above theorem as follows: the design of competitive universal online schemes is a necessary technical hurdle towards resolving the matroid secretary conjecture.

We now proceed with proving Theorem~\ref{thm:sec_to_CRS}. Let $\M=(\E,\I)$, and let $\D \in \Delta(2^\E)$. Recall that an online CRM operates in the following model:  a set of active elements $R \sim \D$ and a random permutation $\Pi$ are (independently) sampled by nature, then $\E$ arrive online in order $\Pi$. When $i \in \E$ arrives, it is revealed whether $i \in R$, and if so the online CRS must determine whether to select $i$. The online CRM must only select an independent subset of $R$.

Suppose we are given a secretary algorithm $\A$ for $\M$ with competitive ratio $\gamma$. Without loss of generality, we assume that $\A$ selects only non-zero weight elements. Consider the following online CRM $\phi_w$ for $\M$, parametrized by a weight vector $w \in \RRp^\E$. When element $i$ arrives, if $i \in R$ then it is presented to $\A$ with weight $w_i$, and if $i \not\in R$ then it is presented to $\A$ with weight $0$. $\phi_w$ selects precisely the elements selected by $\A$.

\begin{lemma}
  For every distribution $\D$, we have \[ \Ex_{R \sim \D}[w(\phi_w(R))] \geq \frac{1}{\gamma} \Ex_{R \sim \D}[\rank_w(R)] \]
\end{lemma}
\begin{proof}
  Condition on the choice of $R$, and let $w'_i=w_i$ if $i \in R$ and $w'_i=0$ otherwise. $\E$ are presented to $\A$ in a uniformly random order, with weights $w'_i$, and $\phi_w(R) \sse R$ is the set of elements selected by $\A$. Since $\A$ is $\gamma$-competitive, it follows that $\Ex[w'(\phi_w(R))] \geq \frac{1}{\gamma} \rank_{w'}(\M)$. Since $w'(\phi_w(R))= w(\phi_w(R))$ and $\rank_{w'}(\M) = \rank_w(R)$, we are done.
\end{proof}
\begin{lemma}\label{lem:weight_based_crm}
  If $\D$ is $\alpha$-uncontentious, then \[ \Ex_{R \sim \D}[w(\phi_w(R))] \geq \frac{1}{\gamma \alpha} \Ex_{R \sim \D}[w(R)] \]
\end{lemma}
\begin{proof}
  Combining the previous lemma with Theorem \ref{thm:characterize_uncontentious} (b).
\end{proof}

Recall that we are assuming for now that we know the $\alpha$-uncontentious distribution $\D$, and we can design an online CRM $\phi_\D$ accordingly. $\phi_\D$ will be a random mixture of the maps $\phi_w$ described above; in particular, we will show that there exists a distribution $\W=\W(\D)$ over weight vectors such that the (randomized) online CRM $\phi_\W$ which samples $w \sim \W$ upfront then invokes $\phi_w$ is an online $\gamma \alpha$-CRM for $\D$.

For each element $i \in \E$, let $x_i = \Pr_{R \sim \D} [i \in R]$. For each weight vector $w$ and $i \in \E$, let $y_i(w) = \Pr_{R \sim \D} [i \in \phi_w(R)]$. For each distribution $\W$ over weight vectors and element $i \in \E$, let $y_i(\W) = \Pr_{R \sim \D}[i \in \phi_\W(R)] = \Pr_{R \sim \D, w \sim \W}[ i \in \phi_w(R)]$. Let $\Y = \set{ y(\W) : \W \in \Delta(\RRp^\E)} \sse [0,1]^\E$ be the family of all inclusion probabilities achievable by some online CRM of the form $\phi_\W$. It is immediate that $\Y=  \convexhull(\set{y(w): w \in \RRp^\E})$, and hence $\Y$ is a convex subset of $[0,1]^\E$.

An online $\alpha\gamma$-CRM for $\D$ of the form $\phi_\W$ exists if and only if $\Y$ intersects with the upwards closed convex set $\frac{x}{\alpha \gamma} + \RRp^\E$. Suppose for a contradiction that this intersection is empty; by the separating hyperplane theorem, this implies that there exists $w \in \RRp^\E$ such that $
\frac{1}{\alpha \gamma} \sum_i w_i x_i > \sum_i w_i y_i$ for all $y \in \Y$. In particular, $\frac{1}{\alpha \gamma} \sum_i w_i x_i > \sum_i w_i y_i(w)$. Since $\sum_i w_i x_i = \Ex_{R \sim \D} w(R)$ and $\sum_i w_i y_i(w) = \Ex_{R \sim \D}[w(\phi_w(R))]$, we get a contradiction with Lemma~\ref{lem:weight_based_crm}. This concludes the proof of the theorem.

%%% Local Variables:
%%% mode: latex
%%% TeX-master: "ucrs"
%%% End:

\section{From Contention Resolution to a Secretary Algorithm?}
\label{sec:CRS_to_sec}
One might hope that online contention resolution is equivalent to the secretary problem on matroids. In particular, does a competitive universal online CRS imply a competitive secretary algorithm? We make partial progress towards this question. In particular, we reduce the secretary problem to online contention resolution on a particular uncontentious distribution derived from the matroid and sample of its elements: the distribution of ``improving elements'', as originally defined by \citet{karger_matroidsampling} for purposes different from ours. 

\begin{definition}\label{def:improving}
  Fix a matroid $\M=(\E,\I)$ with weights $w \in \RRp^\E$, and let $p \in (0,1)$. The random set $R$ of \emph{improving elements} with parameter $p$ is sampled as follows: Let $S \sim \ind_p(\E)$, and let $R=R(S) = \set{ i \in \E : \rank_w(S \union i) > \rank_w(S)}$. Equivalently, $R$ is the set of elements in $\E \sm S$ which are not spanned by higher weight elements in  $S$. Another equivalent definition is $R= \set{i \in \E \sm S: i \in \OPT_w(S \union i)}$.
\end{definition}
The maximum-weight independent subset of the improving elements is $(1-p)$-approximately optimal in expectation:
\begin{fact}\label{fact:improving_approximate}
  Fix a weighted matroid $(\M,w)$, and let $R$ be the random set of improving elements with parameter $p$. Each element of $\OPT_w(\M)$ is in $R$ with probability $1-p$. It follows that $\Ex [ w(R) ] \geq \Ex[\rank_w(R)] \geq (1-p) \rank_w(\M)$.
\end{fact}

Note that the random set $R$ of improving elements does not follow a product distribution. In fact, elements are (weakly) positively correlated in general. This is illustrated by the special case of the $1$-uniform matroid on $[n]$ with weights $w_1 > w_2 > \ldots > w_n$, and $p=1/2$: the distribution of $R$ is as described in Example \ref{ex:1}.  As our main result in this section, we nevertheless show that the random set of improving elements is uncontentious.

\begin{theorem}\label{thm:improving_uncontentious}
  Let $\M=(\E,\I)$ be a matroid with weights $w \in \RRp^\E$, and let $p \in (0,1)$. The random set of improving elements with parameter $p$ is $\frac{1}{p}$-uncontentious.
\end{theorem}

Theorem~\ref{thm:improving_uncontentious} and Fact~\ref{fact:improving_approximate}, taken together, essentially reduce the matroid secretary problem to online contention resolution for the distribution of the random set of improving elements, with one caveat we will discuss shortly. In particular, consider the following blueprint for a secretary algorithm:
\begin{enumerate}
\item Let $S$ be the first $\binom(|\E|,p)$ elements arriving online.
\item Let $R=R(S) \sse \E \sm S$ be a sample of the set of improving elements with parameter $p$.
\item After observing $S$, the elements of $\E \sm S$ arrive online in random order and are presented as such to an online contention resolution algorithm, along with their membership status in $R$. Note that membership in $R$ can be determined ``on the spot'' as required for online contention resolution.\footnote{Technically, a CRM requires that elements of $\E$ --- rather than merely $\E \sm S$ --- be presented in uniform random order along with their membership status in $R$. This is easily accomplished  by appropriately interleaving the elements of $S$ --- none of which are in $R$ --- among the elements of $\E \sm S$.}
\end{enumerate}
Now given a $\beta$-competitive $\alpha$-universal online CRS, we set $p=\frac{1}{\alpha}$ and obtain a $\frac{\beta}{1-p}$-competitive secretary algorithm. 
However, the following caveat prevents us from proving a formal theorem of this form: we cannot provide the online CRS with a complete description of the prior distribution. In particular, the distribution $\D$ of improving elements --- while fully described by the weighted matroid $(\M,w)$ and the parameter $p$ --- can not be fully described to the contention resolution algorithm prior to its invocation, since entries of $w$ are revealed online. As such, we learn both the sample $R \sim \D$ and the distribution $\D$ gradually as elements arrive. An oblivious universal online CRS would resolve this difficulty, but unfortunately we proved in Theorem~\ref{thm:prior_independent_impossible} that such a CRS can not exist even for simple matroids and even offline. A reduction from the matroid secretary problem to contention resolution must therefore require a CRS which can make do with only partial knowledge of the prior. We leave exploration of these possibilities for future work, and discuss them further in the Conclusion section.

\subsection{Proof of Theorem~\ref{thm:improving_uncontentious}}
 
 Let $p$, $S$, and $R$ be as in Definition \ref{def:improving}. We prove that $R$ is uncontentious by leveraging (c) from Theorem~\ref{thm:characterize_uncontentious}. In particular we will show that, for arbitrary $F \sse \E$.
\[ \Ex[\rank(R \intersect F)] \geq p \Ex [ |R \intersect F|]  \]
We break this up into the following three lemmas.

\begin{lemma}
  \[\Ex [\rank(R \intersect F)] \geq (1-p) \Ex[|F \intersect \OPT_w(S \union F)|] \]
\end{lemma}
\begin{proof}
Let $T=S \sm F$, and note that $S \union F = T \uplus F$. We condition on the random variable $T$ and show that the following holds conditionally
\begin{equation}
  \label{eq:1}
    \Ex [\rank(R \intersect F)] \geq (1-p) |F \intersect \OPT_w(T \uplus F)|].
\end{equation}
  
Take $i \in F \intersect \OPT_w(T \uplus F)$. We will show that $i$ is in $R$, and hence is in $R \intersect F$, with probability $1-p$. Since $i \in S \union i \sse T \uplus F$ and $i \in \OPT_w(T\uplus F)$, it follows from the matroid axioms that  $i \in \OPT_w(S \union i)$.  With probability $1-p$ we also have $i \not\in S$, in which case $i \in R$ by definition. 

Since  $F \intersect \OPT_w(T \uplus F)$ is an independent set, \eqref{eq:1} follows.
\end{proof}
%\newpage
\begin{lemma}
 \[|F \intersect \OPT_w(S \union F)| \geq |F \intersect \OPT_w(S)| \]\nopagebreak%
\end{lemma}
\begin{proof}
We prove this by induction on a set $T$ with $S \sse T \sse S \union F$, initialized to $T=S$ at the base case. Consider how the value of $|F \intersect \OPT_w(T)|$ changes as we add elements of $F \sm S$ to $T$ one by one. When adding an element $i \in F \sm T$ to $T$, there are three cases:
\begin{itemize}
\item $i\not\in \OPT_w(T \union i)$: In this case, $\OPT_w(T \union i) = \OPT_w(T)$ and  $|F \intersect \OPT_w(T \union i)| = |F \intersect \OPT_w(T)|$.
\item $i$ is not spanned by $T$, and $i \in \OPT_w(T \union i)$: In this case, $\OPT_w(T \union i) = \OPT_w(T) \union \set{i}$, and therefore $|F \intersect \OPT_w(T \union i)| = 1+ |F \intersect \OPT_w(T)|$.
\item $i$ is spanned by $T$, and $i \in \OPT_w(T \union i)$: In this case, elementary application of the matroid axioms implies that  $\OPT_w(T \union i) = \OPT_w(T) \union \set{i} \sm \set{j}$ for some $j \in T$. Since $i \in F$, it follows that $|F \intersect \OPT_w(T \union i)|$ is either equal to $|F \intersect \OPT_w(T)|$ or exceeds it by $1$, depending on whether $j \in F$.
\end{itemize}
\end{proof}
\begin{lemma}
  \[\Ex[|\OPT_w(S) \intersect F|] \geq \frac{p}{1-p} \Ex [|R \intersect F|] \]
\end{lemma}
\begin{proof}
  For each $i \in F$, we will show that $\Pr[i \in \OPT_w(S)] \geq \frac{p}{1-p} \Pr[i \in R]$, which suffices.

  Take $i \in F$, and let $S_{>i}= \set{j \in S : w_j > w_i}$. Conditioning on $S_{>i}$, there are two cases:
  \begin{itemize}
  \item $i \in \spn(S_{>i})$: It follows that $i \not\in \OPT_w(S)$ and  $i \not\in R$, with certainty.
  \item $i \not\in \spn(S_{>i})$: With probability $p$ we have $i \in S$ and therefore $i \in \OPT_w(S)$ and $i \not\in R$. With the remaining probability $(1-p)$ we have $i \not\in S$ and therefore $i \in R$ and $i \not\in \OPT_w(S)$.
  \end{itemize}
In both cases, the conditional probability that $i \in \OPT_w(S)$ is at least $\frac{p}{1-p}$ times the conditional probability that $i \in R$. The lemma follows.
\end{proof}

\subsection{Where Prior Work Fails}\label{sec:prior_fails}
There has been speculation in the community that contention resolution for improving element distributions can be accomplished online using the ideas of \citet{OCRS}. If this were true, then a stronger (online) form of  our Theorem~\ref{thm:improving_uncontentious} would follow. We show that such conjecture is fatally flawed: there exists no $o(n)$-competitive online CRS in the worst-case arrival model, even when  both the order and the distribution of improving elements  are known to the algorithm. In other words, any competitive online CRS for improving element distributions must make and exploit assumptions on the arrival order. This rules out direct application of the arguments and techniques of \citet{OCRS}, which --- in holding for an (unknown) worst-case arrivals --- cannot exploit the uniform arrival order. The same can be said for the work of \citet{OCRS_prophet}, which operates in the known worst-case arrival model. 

We prove the following theorem, then elaborate on how algorithms from prior work tend to fail on simple examples.

\begin{theorem}
  Let $\M$ be a matroid on $n$ elements. There is no $o(n)$-competitive online CRS for (known) improving element distributions on $M$ in the worst-case arrival model. This holds even for the $1$-uniform matroid, for every constant parameter $p$ of the distribution of improving elements, and even when the arrival order is known to the algorithm.
\end{theorem}
\begin{proof}
Let $\set{1,\ldots,n}$ denote the ground set of of a $1$-uniform matroid, listed in decreasing order of weight. Let $R$ be the random set of improving elements with parameter $p$. Note that $R$ is supported on sets of the form $\set{1,\ldots,k}$ for $k=0,\ldots,n$. In the special case of $p=1/2$, the distribution of $R$ is as described in Example~\ref{ex:1}. In general, $\Pr[ R= \set{1,\ldots,k}] = p (1-p)^k$. The random set $R$ is $1/p$ uncontentious, as shown by Theorem~\ref{thm:improving_uncontentious}. Concretely, the offline CRM $\phi(\set{1,\ldots,k}) = k$ is $\frac{1}{p}$-competitive.

Now suppose that elements are known to arrive online in the order $1,2,3,\ldots,n$, and consider an $\alpha$-competitive online CRM for some $\alpha \geq 1$. Let $T \sse R$ be the (random) set of elements selected by the CRM.  Conditioned on $i \in R$, the CRM must select $i$ with probability at least $\frac{1}{\alpha}$. Formally, $\Pr[ i \in T | i \in R] \geq \frac{1}{\alpha}$.

  When element $i$ arrives, the CRM learns whether $i \in R$, and if so must decide whether to select $i$.  Since the online CRM has only observed elements $1,\ldots,i$, and must make its decision on the spot, it cannot distinguish between different sets of the form $R=\set{1,\ldots,k}$ for $k \geq i$. In other words, it cannot distinguish between the different realizations of $R$ which include $i$, and must therefore select $i$ with probability at least $\frac{1}{\alpha}$ in every realization of $R$ which includes $i$. Formally, $\frac{1}{\alpha} \leq \Pr [ i \in T | i \in R] = \Pr [ i \in T | R=\set{1,\ldots,k}]$ for every $k \geq i$.

Since $i$ was chosen arbitrarily, we can take $k=n$ and conclude that $\Pr[ i \in T | R=\set{1,\ldots,n} ] \geq \frac{1}{\alpha}$ for all $i$. Feasibility requires that $\sum_{i=1}^n \Pr [ i \in T | R = \set{1,\ldots,n}] \leq 1$. Therefore, $\alpha \geq n$.
\end{proof}

It is instructive to examine where the algorithm of \citet{OCRS} fails in the special case of the $1$-uniform matroid on $n$ elements, even when improving elements are presented in a uniformly random order. Indeed, we will argue that no ``simple tricks'' seem to save the day. Recall that the algorithm of \cite{OCRS} defines a sequence of nested flats $\emptyset \subset F_1 \subset F_2 \subset \ldots \subset F_k$, and runs the greedy online algorithm on each contracted submatroid $F_i / F_{i-1}$.  The $1$-uniform matroid contains only a single non-empty flat, containing all elements. Therefore, the algorithm of \cite{OCRS} reduces merely to the naive greedy online algorithm which simply selects the first element it encounters, which in the case of a uniform arrival order is a uniformly random improving element. 

Now, let $[n]=\set{1,\ldots,n}$ denote the ground set of the $1$-uniform matroid listed in decreasing order of weight, and consider the distribution of improving elements $R$ with parameter $p=1/2$ as described in Example~\ref{ex:1}. Element $k$ is improving with probability $2^{-k}$, yet is selected by the algorithm with probability $\sum_{i=k}^{n-1} 2^{-(i+1)} \cdot \frac{1}{i} + 2^{-n} \cdot \frac{1}{n} < \frac{2^{-k}}{k} = \frac{\Pr [ k \in R]}{k}$. Intuitively, when $k$ is improving, so are elements $1,\ldots,k-1$, which easily span $k$ and are not distinguished from $k$ by the algorithm. It is easy to show that the algorithm suffers the same fate for any other choice of $p$. 

One might be tempted to employ other tricks, such as for example ``canceling'' each element in $R$ with independent constant probability in order to reduce contention and place the marginal probability vector deep in the matroid polytope. Such tricks are doomed to fail all the same: the algorithm groups all elements into the same (unique) flat, and in doing so does not distinguish between the ``uncanceled'' elements of $R$, so cannot select element $k$ with probability exceeding $\frac{\Pr[ k \in R]}{k}$.

It is hopefully now clear that any online CRS for improving element distributions must make and exploit assumptions on the arrival order. Whereas this rules out obvious extensions of \citet{OCRS} and \citet{OCRS_prophet}, one might hope that the algorithm of \citet{ROCRS} might fare better, since they do exploit the random ordering assumption. Sadly,  their algorithm also fails for the $1$-uniform matroid: it also does not distinguish between different improving elements in this special case, and therefore also selects element $k$ with probability no more than $\frac{\Pr[ k \in R]}{k}$. That being said, we are more hopeful that the techniques of \cite{ROCRS}, if combined with significant new ideas, might yield progress on online contention resolution for positively correlated distributions.

%%% Local Variables:
%%% mode: latex
%%% TeX-master: "ucrs"
%%% End:

\section{Conclusions and Open Problems}

%Oblivious / prior indep for product dists?
% % Can we show a universal online scheme unconditionally? Remove requirement of matsec... Perhaps a construction like that of OCRS might work, but doesnt seem to work as far as we can tell (maybe try and see where it breaks?)
% %Is there an oblivious universal (offline or online) scheme for improving dists? Would imply matsec if online. Might matsec also imply it?
% %We went into this looking to understand relationship, hoping for some kind of equivalence. Fell short. Is there a tighter connection to be made here to between CR and matsec? How do we define partially oblivious scheme? The one that is needed for improving dists to imply matsec, and might be implied by matsec as well?
% Can section 5 result be made computationally efficient

In this paper, we begin an exploration of the power and limitations of contention resolution beyond known product distributions, as well as its connections to secretary problems. We hope that our results are a first step towards broader application of the techniques behind contention resolution and online selection. Most notably, our results highlight approaches to resolving the matroid secretary conjecture. We identify several intriguing open questions in pursuit of these agendas.
\begin{itemize}
\item Can the result of Theorem~\ref{thm:sec_to_CRS} be shown unconditionally; i.e., can we show a competitive universal online CRS for matroids without assuming the matroid secretary conjecture? We believe this to be a reasonable first step towards proving the matroid secretary conjecture. As we show in Section~\ref{sec:prior_fails}, prior work on online contention resolution fails in the presence of even the modest positive correlation exhibited by (uncontentious) improving element distributions on simple matroids. Therefore, we believe significant  new ideas are required.
\item Recalling the caveat to our results from Section~\ref{sec:CRS_to_sec}, can a tighter connection be made between the secretary problem and contention resolution? Is there a natural model of contention resolution on matroids which permits a reduction both from and to the matroid secretary problem? The knee-jerk approaches using duality-like arguments fail to establish such an equivalence, so new ideas appear to be required.
\item The caveat to our results from Section~\ref{sec:CRS_to_sec} suggests that resolving contention with limited knowledge of the prior is closely related to the matroid secretary conjecture. Recalling our impossibility result of Theorem~\ref{thm:prior_independent_impossible}, we can start by examining prior-independent contention resolution for interesting classes of distributions.  For example,  is there a competitive prior-independent (or even oblivious) CRS for ex-ante-feasible product distributions? 
\item Can Theorem~\ref{thm:sec_to_CRS} be made computationally efficient? Given only oracle access to an arbitrary uncontentious distribution and an arbitrary algorithm for the matroid secretary problem, this is unclear.
\item Is there an analogue of our characterization of uncontentious distributions for prophet inequality problems? In particular, can we characterize joint distributions of random variables which permit competitive prophet inequalities with respect to a given matroid?
\item Do more general set systems permit a characterization of uncontentious distributions with a finite set of inequalities, a-la Theorem~\ref{thm:characterize_uncontentious}?
\end{itemize}

We restricted our attention to matroids in the paper, though some notes are in order on extensions of our results to more general constraints. In the characterization of Theorem~\ref{thm:characterize_uncontentious}, the equivalence of (a) and (b) holds for a general downwards-closed set systems, and is implicit in the arguments of \cite{CRS}. The equivalence with (c) exploits the matroid structure, however. Theorem~\ref{thm:sec_to_CRS} also holds for general downwards-closed set systems, and our proof does not invoke the matroid assumption. The results and arguments of Section~\ref{sec:CRS_to_sec}, in particular Theorem~\ref{thm:improving_uncontentious}, heavily rely on the matroid structure and do not appear to be easily extensible beyond matroids. We leave further extensions of our results beyond matroids for future work.

%%% Local Variables:
%%% mode: latex
%%% TeX-master: "ucrs"
%%% End:

%\newpage

{%\small
\bibliography{ucrs}
\bibliographystyle{abbrvnat}               %latex8
}

%\newpage
%\appendix
%\input{appendix}

\end{document}